\documentclass[letterpaper, 10 pt, conference]{ieeeconf}  

\IEEEoverridecommandlockouts                              

\usepackage{graphicx} 
 \usepackage[caption=false,font=footnotesize]{subfig}
\usepackage[cmex10]{amsmath} 
\usepackage{amssymb}  
\usepackage{color}
\usepackage{bbm}
\usepackage{algorithm,algorithmicx,algpseudocode}
\usepackage{url}

\usepackage{float}

\usepackage{epstopdf}

\usepackage{tikz}
\usetikzlibrary{automata,positioning}

\makeatletter

\newcommand\floatc@myruled[2]{{\@fs@cfont #1} #2\par}
\newcommand\fs@myruled{\def\@fs@cfont{\bfseries}\let\@fs@capt\floatc@myruled
\def\@fs@pre{\hrule height.8pt depth0pt \kern2pt}%
\def\@fs@post{\kern2pt\vspace{-10pt}\hrule\relax}%
\def\@fs@mid{\kern2pt\hrule\kern2pt}%
\let\@fs@iftopcapt\iftrue}
  
\floatstyle{myruled}
\newfloat{Problem}{tbp}{prob}

\usepackage[official]{eurosym}

\usepackage{mathtools}

\newtheorem{theorem}{Theorem}
\newtheorem{lemma}{Lemma}

\title{\LARGE \bf
Relocation in Car Sharing Systems with Shared Stackable Vehicles: Modelling Challenges and Outlook
}

\author{Chiara Boldrini, Riccardo Incaini, Raffaele Bruno \\
IIT-CNR \\
Via G. Moruzzi 1, 56124, Pisa, ITALY \\
{\tt\small \{first.last\}@iit.cnr.it}
\thanks{*This work was partially funded by the ESPRIT project. This project has received funding from the \emph{European Union's Horizon 2020 research and innovation programme} under grant agreement No 653395. This work was also partially funded by the REPLICATE project. This project has received funding from the \emph{European Union's Horizon 2020 research and innovation programme} under grant agreement No 691735.}
}

\begin{document}

\maketitle
\thispagestyle{empty}
\pagestyle{empty}

\begin{abstract}
Car sharing is expected to reduce traffic congestion and pollution in cities while at the same time improving accessibility to public transport. However, the most popular form of car sharing, one-way car sharing, still suffers from the vehicle unbalance problem. Innovative solutions to this issue rely on custom vehicles with stackable capabilities: customers or operators can drive a train of vehicles if necessary, thus efficiently bringing several cars from an area with few requests to an area with many requests. However, how to model a car sharing system with stackable vehicles is an open problem in the related literature. In this paper, we propose a queueing theoretical model to fill this gap, and we use this model to derive an upper-bound on user-based relocation capabilities. We also validate, for the first time in the related literature, legacy queueing theoretical models against a trace of real car sharing data. Finally, we present preliminary results about the impact, on car availability, of simple user-based relocation heuristics with stackable vehicles. Our results indicate that user-based relocation schemes that exploit vehicle stackability can significantly improve car availability at stations. 
\end{abstract}

%
%
\section{Introduction}
\label{sec:intro}
\noindent
Car sharing is considered one of the pillars of the smart mobility infrastructure for smart cities. With car sharing, car access is decoupled from car ownership: people do not own a car, they simply rent it from the car sharing operator when they need it (typically for short-range trips), effectively implementing the concept of Mobility-as-a-Service. In cities where car sharing services are running, positive effects have already been measured: car sharing members use cars less, rely more on public transport or bicycles, and in some cases they even shed their private car (or refrain from buying a second one for their family)~\cite{martin2016impacts}. Car sharing can also act as a last-kilometre solution for connecting people with public transport hubs, hence becoming a feeder to traditional public transit~\cite{shaheen2001commuter}.

One-way car sharing, in which customers are not forced to return the vehicle at the starting point of their journey, is the most popular among customers, due to the great flexibility it provides. One-way systems can be also classified into \emph{free-floating} or \emph{station-based} according to their parking restrictions. In fact, the former refers to a system in which the return of the rented vehicle is possible at any parking spot within the operational area of the car sharing service\footnote{For example, \url{https://www.car2go.com/}.}, while the latter requires that pickups or drop-offs of vehicles occur at designed parking stations deployed by the operator\footnote{For example, \url{https://www.autolib.eu/en/}.}. Advantages of station-based systems is that they ensure higher reliability and predictability of car locations and parking, which make possible advance reservations from the users. One-way car sharing is not without drawbacks for the car sharing operators.  With one-way car sharing, cars will follow the natural flows of people in a city, hence accumulating in commercial/business areas in the morning and in residential areas at night~\cite{boldrini16characterising}. As a result, the availability of cars can become extremely unbalanced during the day, and certain areas may end up being underserved due to lack of available cars. 

Previous research has proposed several approaches to solve the vehicle unbalance problem, including: user-based relocation, i.e., price incentives for the users to relocate the vehicles themselves~\cite{trr12_relocation}; operator-based relocation, i.e., workforce that moves vehicles from where they are not needed to where there is a significant demand~\cite{Kek2009,Boyaci2015}; and optimal planning of station deployment to achieve better service accessibility and a more favourable distribution of vehicles~\cite{itcs16_biondi}. It is important to point out that \emph{the relocation process is intrinsically inefficient}: as one driver per car is needed, to relocate several cars a large workforce or many willing customers are necessary. This significantly complicates the relocation with respect to, e.g., bike sharing services, where a single worker with a van can redistribute a large amount of bicycles. 

In order to address the above limitations, it is fundamental to reduce the ratio between the workforce size and the number of vehicles that can be relocated. Recently, innovative technologies have been proposed to enable more efficient vehicle rebalancing in CS systems. On the one hand, the rebalancing problem is solved in~\cite{ijrr12_fluid} using empty robotic vehicles autonomously driving between stations. On the other hand, new vehicle concepts with \emph{stackable capabilities} have been recently released or are under development, which can be stacked into a train (through a mechanical and electric coupling) and/or folded together. Then, the train can be driven either by a car sharing worker (up to 8 vehicles) or by a customer (up to 2 vehicles). An illustration of this type of vehicle prototyped in the ESPRIT project~\cite{esprit} is provided in Fig.~\ref{fig:esprit}. Such stackable cars come with the promises of significantly improving the system manageability of future car sharing services. However, the evaluation and design of new car sharing services using these innovative stackable cars call for new modelling techniques able to accurately characterise their peculiar properties. 
\begin{figure}[th]
\begin{center}
\includegraphics[scale=0.4]{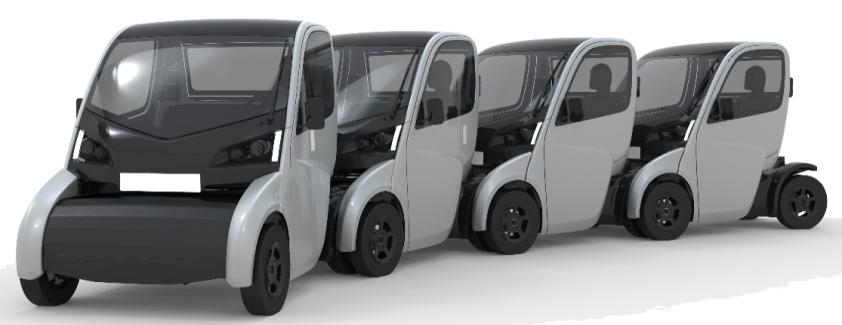}
\caption{The ESPRIT train of vehicles}
\label{fig:esprit}
\end{center}\vspace{-15pt}
\end{figure}

Various models for assessing the performance of one-way car sharing systems have been proposed in the literature. A class of modelling approaches relies on \emph{time-space} models that describe the interactions between the operational decisions (i.e., movement of staff, relocation activities) and the number of vehicles at the station~\cite{Kek2009,Boyaci2015}. The main drawbacks of this approach is the explosion in size of the state space, and the limited ability to deal with uncertain conditions due to the stochastic nature of customer arrivals. Thus, stochastic models have been recently proposed based on queueing theoretical approaches~\cite{or11_queue,rr16_queue} or fluid approximations~\cite{ijrr12_fluid}. However, how to model car sharing systems with stackable vehicles is still an open challenge. 

Within this framework, the contributions of this paper can be summarised as follows:
\begin{list}{\tiny$\bullet$}{\leftmargin=1em \itemindent=-0.5em}
\item We validate a queueing theoretical model of one-way car sharing systems that was proposed in prior work~\cite{or11_queue} against a trace of real car sharing operational data, confirming the effectiveness of this modelling approach;
\item We show, with a real-case study, that increasing the fleet size cannot solve the problem of vehicle shortages in hot spots because vehicle availability depends mostly on traffic patterns in the car sharing system;
\item We present preliminary results about two simple heuristics for user-based relocation with stackable cars, showing that in some cases user-based relocation can increase the car availability at stations up to 200\%;
\item We develop a new queuing model to study the evolution of vehicle redistribution in a car sharing station under general user-based relocation policies for stackable cars and we derive an upper-bound on the relocation flow per station.
\end{list}
\section{Preliminaries on Queueing Network Models for CS Systems}\label{sec:queueing_model}
\noindent
In this section, we recall the queueing theoretical model proposed in the literature by~\cite{or11_queue}. To this aim, we assume that the CS operational area is composed of a set $\mathcal{S}$ of non-overlapping car sharing \emph{service centres}. A service centre abstracts what, concretely, can be a CS station (in station-based car sharing) or a zone in a free-floating car sharing operational area. Hence, the queueing theoretical model can be applied to both types of one-way CS systems. Without ambiguity, in the following we will use the terms ``service centre'' and ``station'' interchangeably. At each service centre, shared cars are dropped off at the end of a journey and are picked up by other customers starting their new journeys. We assume that the inter-arrival time between customers that pick up cars at station $i$ is exponentially distributed with rate~$\mu_{i}$ and that, similarly to~\cite{or11_queue}, customers simply leave the service centre if they cannot find an available vehicle. 
In order to keep the analytical model tractable, in this work, similarly to the related literature, we assume that the capacity of service centres is not particularly critical, hence we neglect potential losses due to a service centre being fully occupied at the time a newly dropped-off car arrives.
Based on the above assumptions, the service centre can be represented, in Kendall's notation, as a -/M/1 queue~\cite{bolch2006queueing}. 

Individual queues representing the CS service centres are then networked together to reflect the CS network dynamics. To this aim, the probability matrix $\mathbf{P}= \{ p_{ij} \}_{i,j \in \mathcal{S}}$ is introduced, whereby $p_{ij}$ denotes the probability that a customer leaving service centre $i$ with a car will head for service centre~$j$. In addition, in order to model the fact that it takes a certain amount of time to go from service centre~$i$ to service centre~$j$, we introduce, as in~\cite{or11_queue}, delay queues in the model. These delay queues are modelled as infinite-servers queues, and we denote their set with $\mathcal{I}$. There will be a delay queue between service centres $i$ and $j$ if $p_{ij} \neq 0$, thus  $| \mathcal{I}| \leq |\mathcal{S} \times \mathcal{S}|$. Each server in an infinite-servers queue keeps a car for an exponential amount of time with mean $T_{ij}$ (where $T_{ij}$ is the expected travel time from $i$ to $j$)\footnote{The model does not consider traffic congestion, thus each vehicle travels (i.e. is served) independently and in parallel with the others. This implies that the overall service rate of the delay queue is proportional to the number of vehicles in the queue, as stated in Equation~\ref{eq:service_rates}.}. 

We can then summarise the characteristics of the CS queueing network as follows. We denote the number of shared cars in the CS system with $N$ and we let $\mathcal{K} = \mathcal{S} \cup \mathcal{I}$.  The service rate of each queue $i \in \mathcal{K}$ is given by:
\begin{equation}\label{eq:service_rates}
\mu_{i}(n_i) = \begin{cases} \mu_{i} & \textrm{if } i \!\in\! \mathcal {S} \\
\frac{n_i}{T_{jk}} & \textrm{if } i \in \mathcal{I}, j=o(i), k=d(i) 
\end{cases}
\end{equation}
where $n_i\in \{0,1,\ldots, L\}$ denotes the number of vehicles at node~$i$ and, for each $i \in \mathcal{I}$ $o(i)$ and $d(i)$ denote the upstream and downstream service centres of delay queue $i$ (corresponding to the origin and destination service centres of the CS trip). The routing matrix, which describe how vehicles move across queues, can then be obtained as follows:
\begin{equation}
r_{ij} = \begin{cases} p_{il} & \textrm{for } i\in \mathcal{S}, j\in \mathcal{I}, i=o(j), l=d(j),  \\
1 & \textrm{for } i\in \mathcal{I}, j\in \mathcal{S}, j=d(i), \\
0 & \textrm{otherwise} 
\end{cases} \; . 
\end{equation}
The queueing network described above belongs to the category of single-class closed queueing networks, in particular it is a BCMP network~\cite{bolch2006queueing}. 
For the sake of illustration, we provide a simple example in Fig.~\ref{fig:queueing_network}.

\begin{figure}[tb] 
\centering
\includegraphics[trim={0cm 0cm 0cm 0cm},clip,angle=0,width=0.45\textwidth]{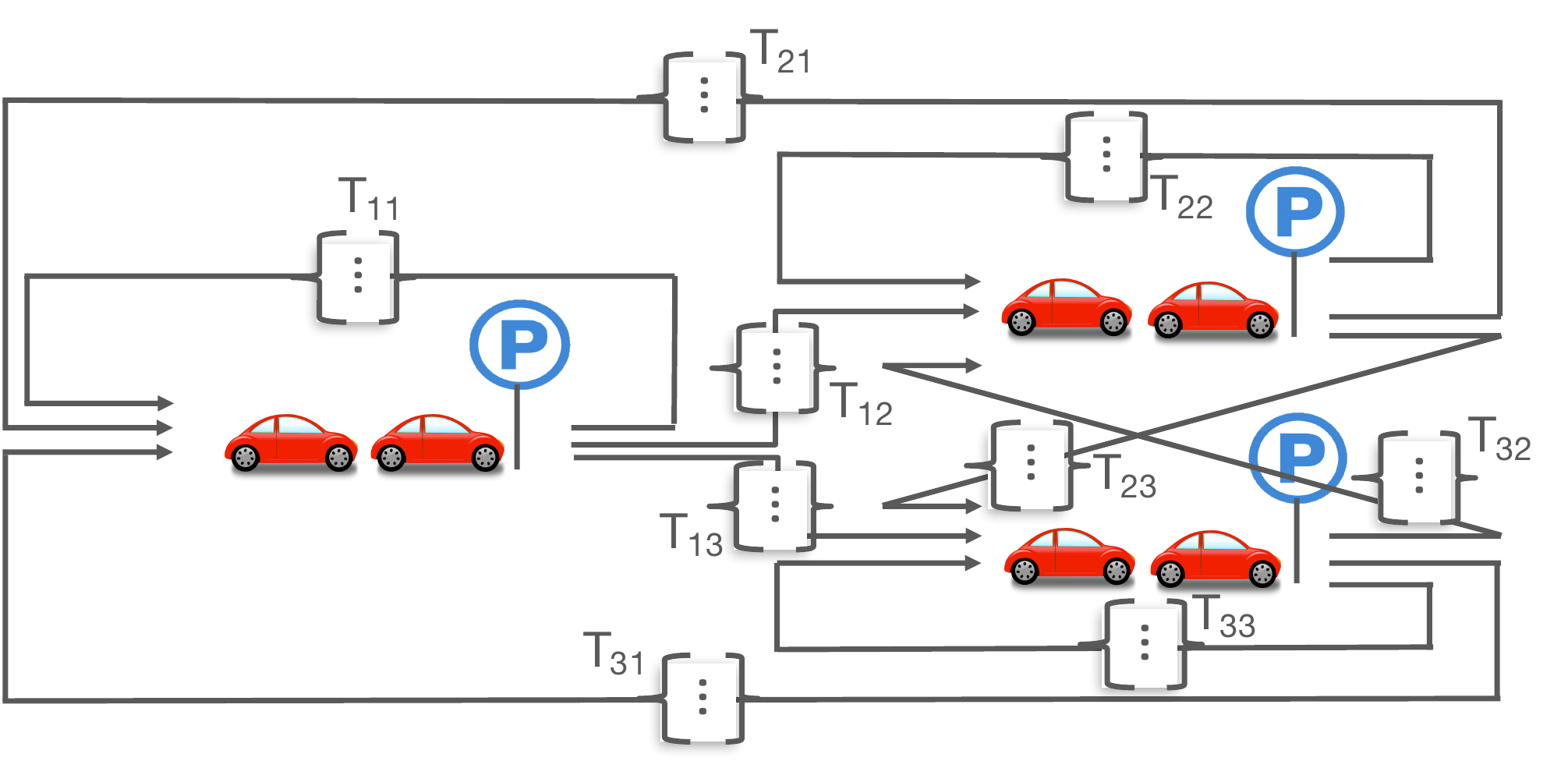}\vspace{-0.3cm}
\caption{Model of a CS system with 3 stations.\label{fig:model}}
\label{fig:queueing_network}
\vspace{-0.5cm}
\end{figure}

The first step in solving this model consists in solving the traffic equations~\cite{bolch2006queueing}, which in our case simplify to the following:
\begin{equation}
\begin{cases}
e_i=\sum_{j \in S} e_{j} p_{ij} & \forall i \in \mathcal{S} \\
e_{ij}=e_i p_{ij} & \forall i,j \in \mathcal{S} \\
\end{cases},
\end{equation}
where $e_i$ denotes the relative arrival rate at queues corresponding to CS stations and $e_{ij}$ the relative arrival rate to the delay queue linking station $i$ and station $j$. 
%
We can now exploit the results for BCMP networks, which are known to have a product form solution for the stationary distribution as follows:
\begin{equation}
P\left(\{n_i\}_{i \in \mathcal{S}}, \{n_j\}_{i \in \mathcal{I}} \right) = \frac{1}{G(N)} \prod_{i \in \mathcal{S}} \left(\frac{e_i}{\mu_i}\right)^{n_i} \!\! \prod_{j \in \mathcal{I}} \left(\frac{e_j}{\mu_j}\right)^{n_j} \!\! \frac{1}{n_i !},
\end{equation}
where $G(N)$ is a normalisation constant. $G(N)$ can be efficiently derived as described in~\cite{bolch2006queueing} using the convolution algorithm. 
Important performance measures can then be obtained exploiting the normalisation constant as follows. The throughputs of both single-server queues (the CS service centres) and infinite-server queues (the delay queues) are given by:
\begin{equation}
\lambda_i = e_i \frac{G(N-1)}{G(N)}, \quad i \in \mathcal{K} \;.
\end{equation}
The throughput at CS service centres corresponds to the intensity of drop-offs. For single server queues, the average number of cars $\overline{N}_i$ parked at the service centres and the utilisation $\rho_i$ can be obtained as:
\begin{equation}
\overline{N}_i = \sum_{n = 1}^{N} \left(\frac{e_i}{\mu_i}\right)^n \frac{G_{i}(N-n)}{G(N)},
\end{equation}
\begin{equation}
\rho_i = \frac{e_i}{\mu_i} \frac{G(N-1)}{G(N)},
\end{equation}
where $G_{i}(N-n)$ is the normalisation constant computed removing queue $i$ and considering $N-n$ jobs. Please note that the utilisation $\rho_i$ is a strategic metric for a car sharing network, since it gives the probability that there is at least one car available for pick up at the station.

%
\section{CS model validation}\label{sec:validation}
\noindent
The model described in the previous section has been already used in~\cite{or11_queue} for fleet sizing and in~\cite{rr16_queue} for deriving a strategy to rebalance vehicles. However, its modelling power has never been validated against real traces. Thus, the model could possibly be inadequate to represent the complexities of real car sharing systems. Hence, before extending it in Section~\ref{sec:relocation_queue} to account for stackable vehicles, we believe it is of paramount importance to first check its validity. To this aim, we rely on a dataset composed of all the pickup and drop-off events of 349 shared vehicles of a free-floating car sharing service operated in The Netherlands. Data is collected every minute for a period of one month and a half between May and June 2015 using openly available APIs. The dataset comprises more than 51,000 trips. Each observation reports the type of the event (pickup or drop-off), the time, the geographical coordinates and the status of the vehicle. No information is available on the trip trajectory. 

Following the service centres strategy discussed in Section~\ref{sec:queueing_model}, we have partitioned the study area into non-overlapping square cells. Each of these cells is then modelled as a -/M/1 queue. To this aim, we need to estimate their service rate $\mu_i$. We use the technique described in~\cite{clarke1957maximum}, whereby the service rate $\mu_i$ for queue $i$ is obtained as $\mu_i = \frac{n_{dep} - n_{init}}{T_{busy}}$, where $n_{dep}$ denotes the number of departures observed at the queue (corresponding to the number of pickups in our terminology), $n_{init}$ is the initial size of the queue (i.e., how many cars were parked at time $t_0$), and $T_{busy}$ is the time the queue has been busy (i.e., with at least one car parked). These quantities can be easily computed from the trace, and their distribution (in log-log scale) is shown in Fig.~\ref{fig:mu_sidelength} for varying cell side length. For smaller cell side lengths, we observe several orders of magnitude of variability in the service rates, owing to the fact that cells are small and there can be very popular ones and very neglected ones. Vice versa, the behaviour is more homogeneous when cells are larger. 

\newcommand\plotextension{pdf}

\begin{figure}[t]
\begin{center}
\includegraphics[scale=0.4]{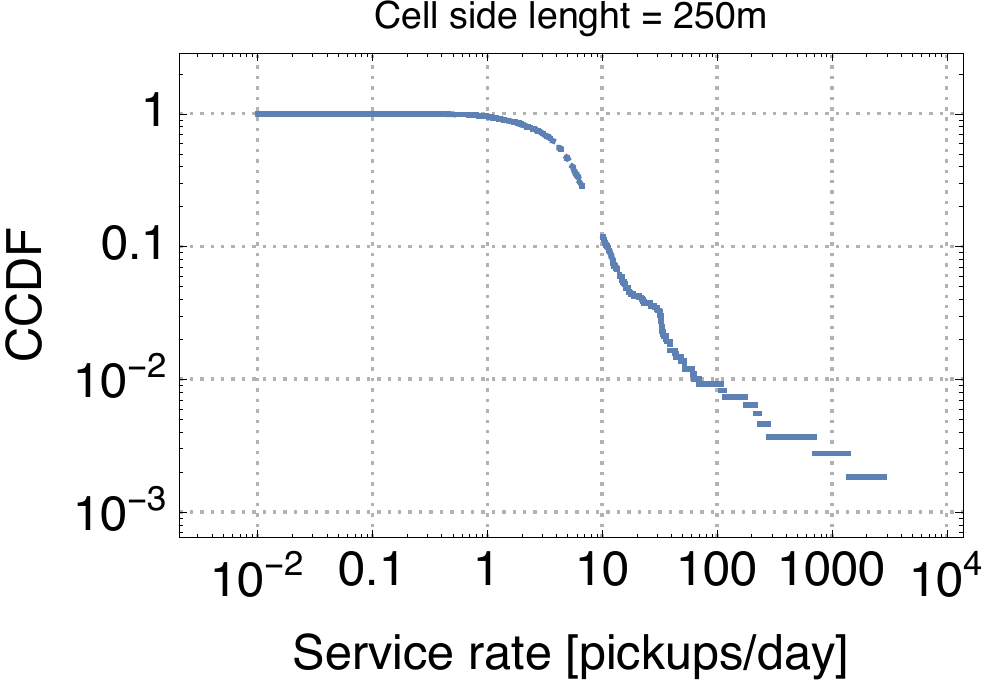}
\includegraphics[scale=0.4]{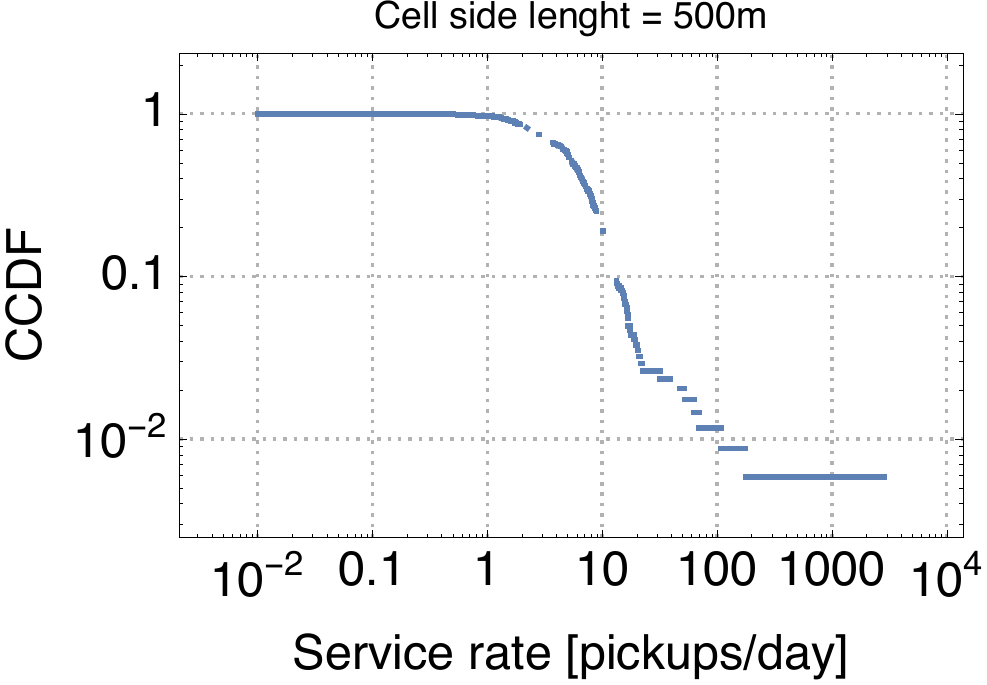}\vspace{5pt}
\includegraphics[scale=0.4]{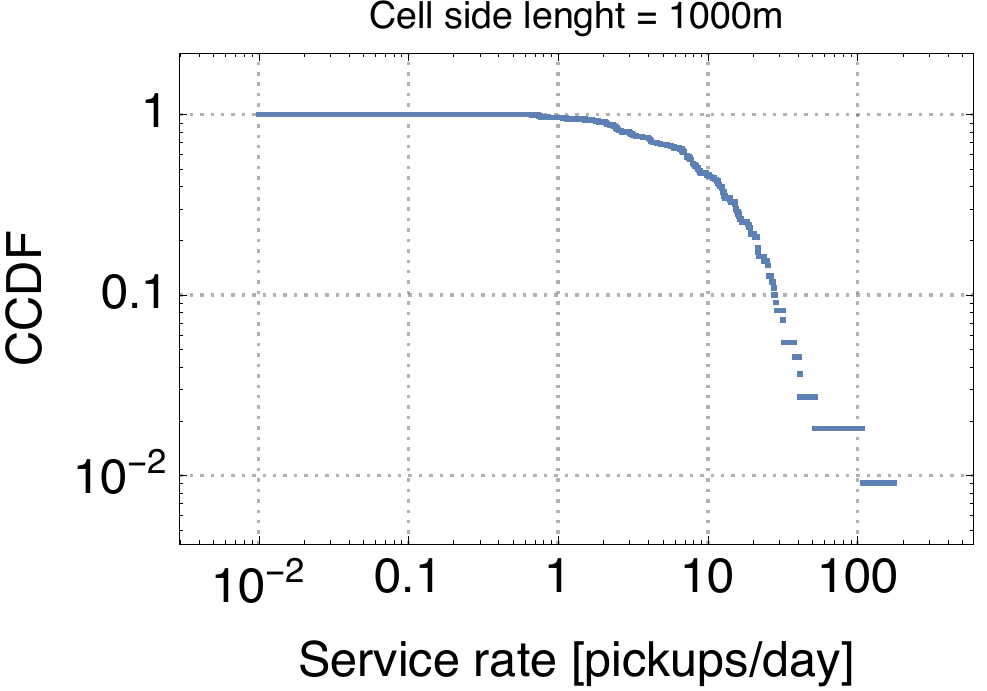}
\includegraphics[scale=0.4]{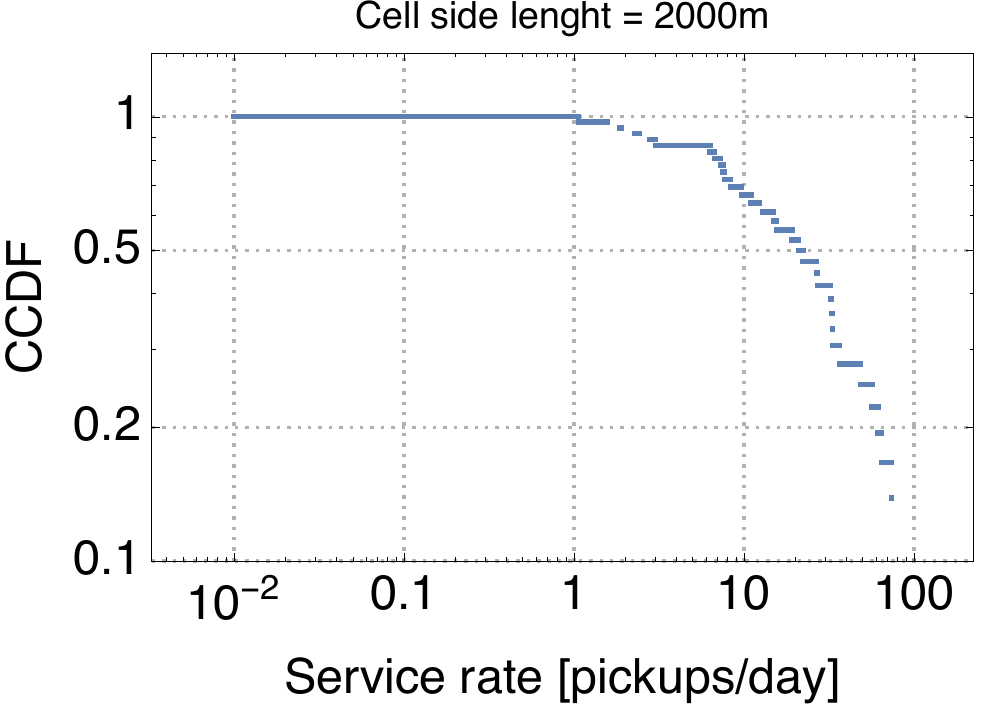}\vspace{-5pt}
\caption{CCDF of service rate measured in the trace}
\label{fig:mu_sidelength}
\end{center}\vspace{-25pt}
\end{figure}

We next feed the service rates $\mu_i$ to the closed network model that we have described in Section~\ref{sec:queueing_model}. We set the number $N$ of vehicles to $349$, as in our trace. The number of cells (i.e., the size of $\mathcal{S}$) is given by the number of active cells for each cell side length configuration. Then we derive the routing probabilities $p_{ij}$ by simply counting, for each service centre $i$, the fraction of trips from $i$ to any service centre $j$. For all the pairs of service centres for which the routing probability is non zero, we also compute the average duration $T_{ij}$ of trips from $i$ to $j$. The service rate of each server in the delay queue for $(i,j)$ is then given by the inverse of $T_{ij}$ (the MLE estimator of the rate of an exponential random variable is simply the inverse of the sample mean). 
For all cell side lengths, the resulting Markov chain is ergodic, hence a unique steady-state probability exists~\cite{bolch2006queueing}.

Since the closed queueing network is completely defined based on the car sharing trace, we can apply the formulas for the throughput, utilization, and average number of cars at the stations that we have derived in Section~\ref{sec:queueing_model}.
%
Due to the lack of space, we show only the results for cell side lengths 250m and 1000m (the results for the other cases are analogous). We observe that for the throughput and availability (Fig.~\ref{fig:throughput_sidelength} and~\ref{fig:availability_sidelength}) the predictions of the theoretical model are quite accurate, regardless of the size of the cell. In particular, it seems that predictions are only offset by a proportionality constant. A less accurate match is obtained for the average number of cars parked at the service centres when the utilisation of service centres is high, but the model still captures pretty closely the general trend of this metric (Fig.~\ref{fig:rhoVScarsparked_sidelength}). Thus, overall, we can conclude that queueing-theoretical approaches can be safely used for modelling CS systems.


\begin{figure}[t]
\begin{center}
\includegraphics[scale=0.4]{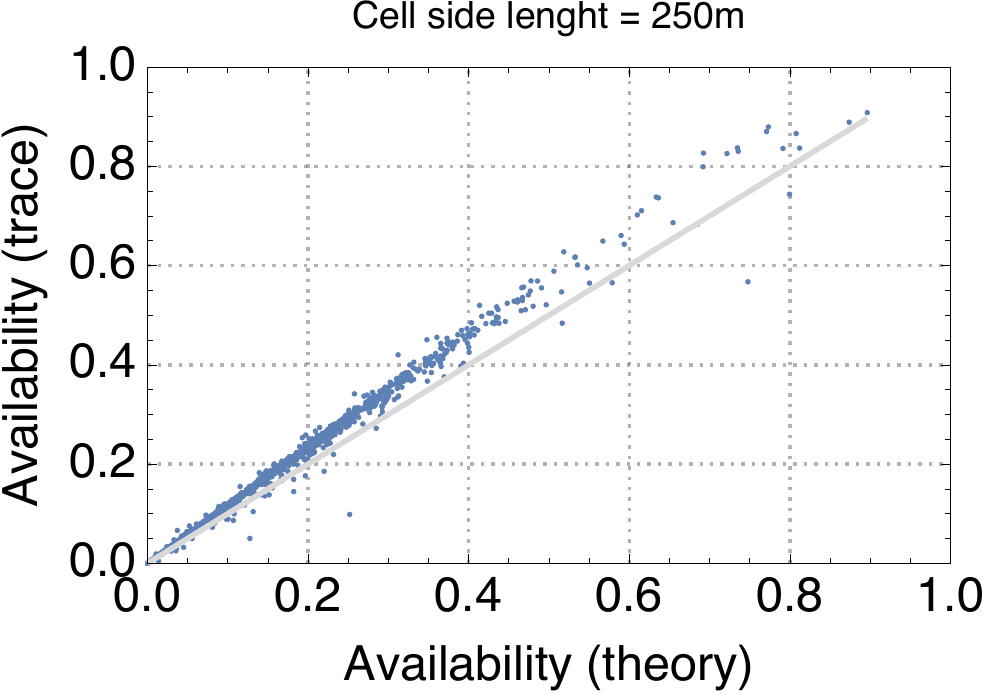}
\includegraphics[scale=0.4]{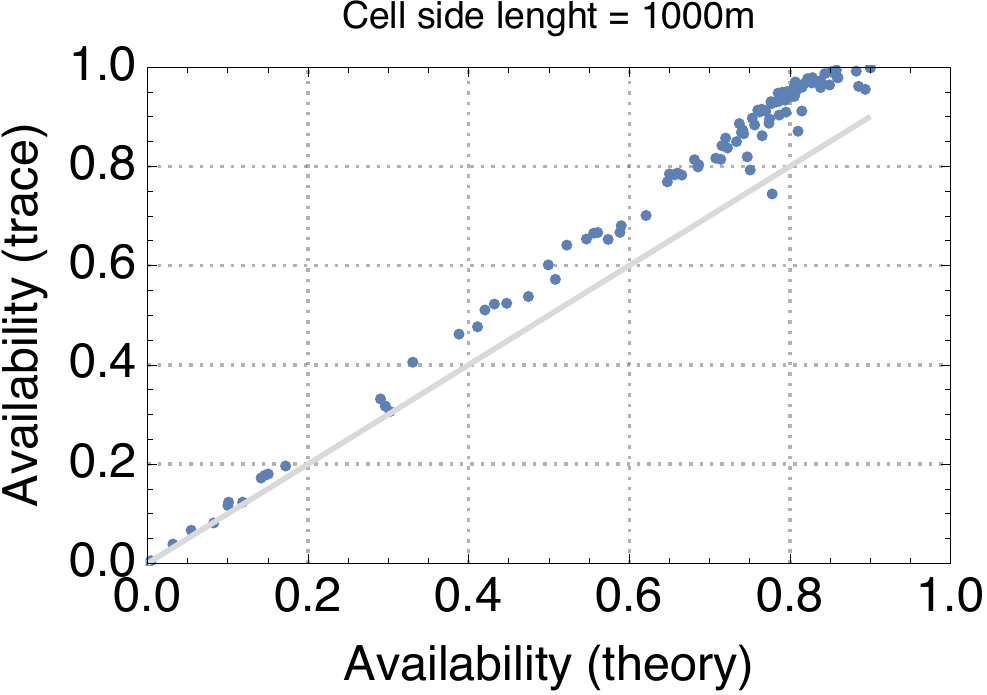} \vspace{-5pt}
\caption{Availability scatterplot (theory vs trace)}
\label{fig:availability_sidelength}
\end{center}\vspace{-10pt}
\end{figure}

 \begin{figure}[t]
\begin{center}
\includegraphics[scale=0.4]{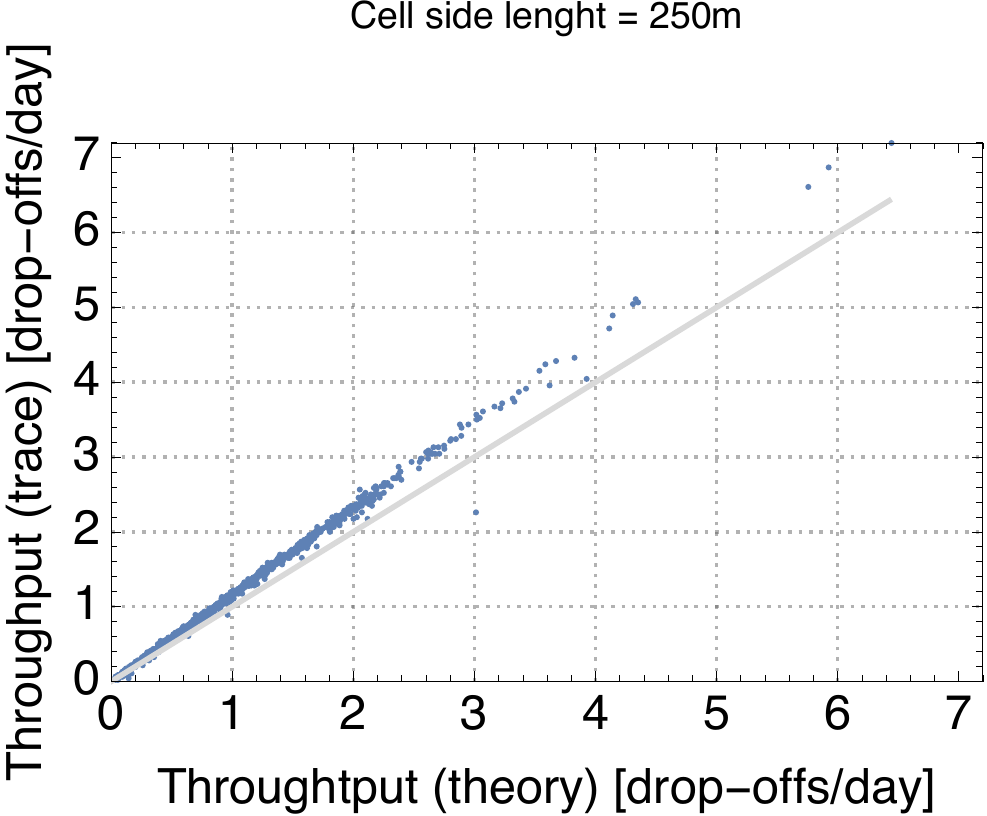}\hspace{5pt}
\includegraphics[scale=0.42]{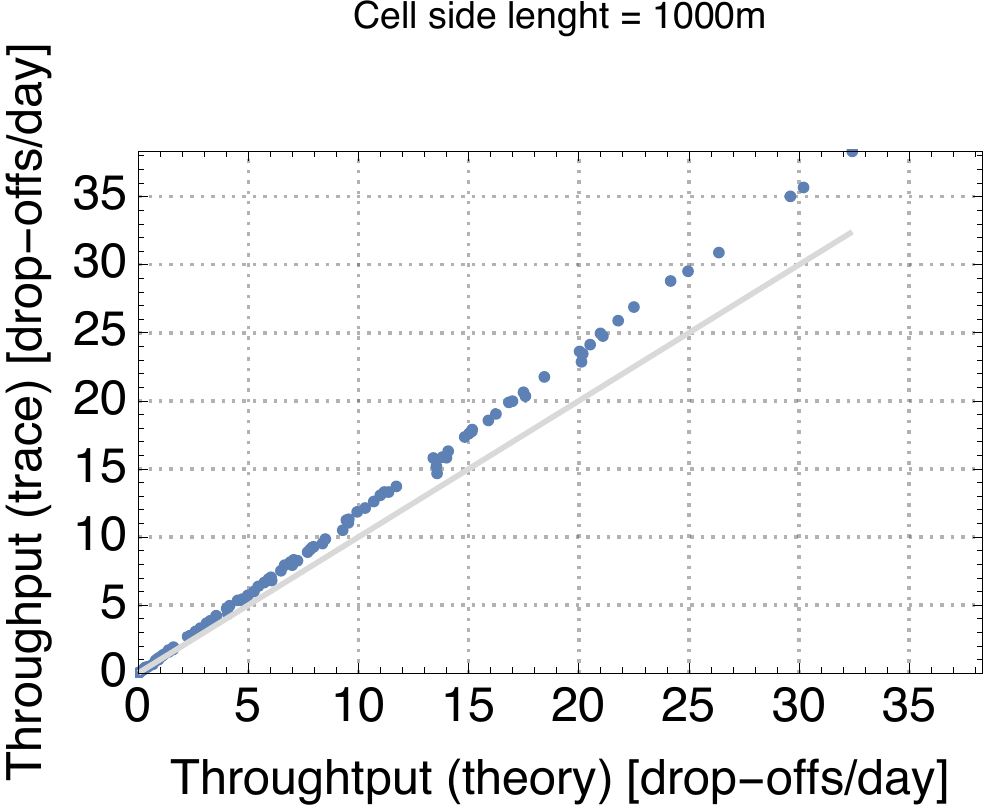} \vspace{-5pt}
\caption{Throughput scatterplot (theory vs trace)}
\label{fig:throughput_sidelength}
\end{center}\vspace{-15pt}
\end{figure}

%

 \begin{figure}[t]
\begin{center}
\includegraphics[scale=0.45]{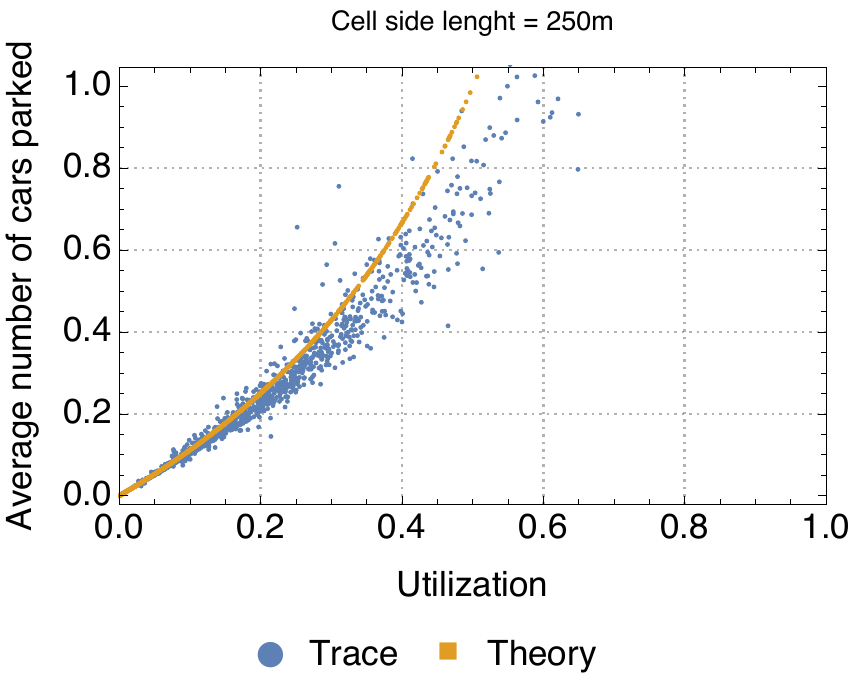}
\includegraphics[scale=0.45]{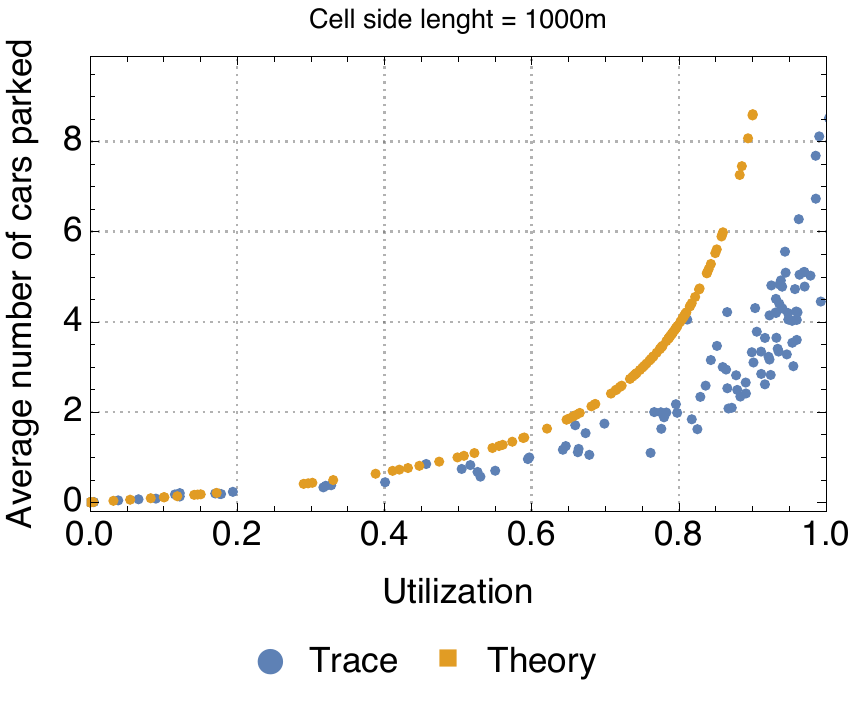}\vspace{-10pt}
\caption{Average number of cars available at stations as a function of the utilization of the station}
\label{fig:rhoVScarsparked_sidelength}\vspace{-10pt}
\end{center}
\end{figure}


%
%
\subsection{Impact of Fleet Size on CS Performance}\label{sec:fleet}
\noindent
It is well known that vehicle availability is influenced by the number, location and size of CS stations~\cite{itcs16_biondi}, as well as by the fleet size~\cite{smc14_fleet,Boyaci2015}. Typically, the optimal fleet size is chosen in such a way to optimise the trade-off between fleet costs and increased revenues due to higher availabilities. In this section, we want to demonstrate with a real-case study that increasing the fleet size, even in the ideal case of unconstrained capital investment, is not the panacea for ensuring higher vehicle availabilities. 
In the literature, George and Xia~\cite{or11_queue} have already shown that increasing the fleet size pumps up the availability at service centres but only until the centre(s) with the highest utilisation become saturated. This would not be a problem if all service centres were homogeneous, i.e., if they had comparable service rates. In that case, increasing the fleet size could bring the system to a situation of maximum availability at all stations. However, as Fig.~\ref{fig:mu_sidelength} shows, the service rates in a real car sharing system tend to be quite heterogeneous. 

In order to illustrate how to exploit the model in Section~\ref{sec:queueing_model} for ``what-if'' studies, in the following we use the set of stations and delay queues obtained in Section~\ref{sec:validation} for cell side length 250m and we test what would happen if the CS operator increased the fleet from 349 to 500 and 1000 shared vehicles, or if it reduced the fleet to 200 shared cars. Owing to the heterogeneity of service rates and the result by~\cite{or11_queue}, we do not expect a significant improvement in the availability of vehicles as the fleet size increases. This is confirmed by Fig.~\ref{fig:availabilityVSstat_fleetsize}, which shows the availability for different fleet sizes at the different stations. With 1000 shared vehicles, our reference car sharing system reaches maximum availability, but only for a very small subset of stations. All the others are left lagging behind. What is worse, those stations that see their availability increase significantly are those whose availability was already higher (in Fig.~\ref{fig:availabilityVSstat_fleetsize} stations are sorted by increasing availability when $N=200$ and the same order is kept for all other $N$ values). Thus, if we increase the fleet size we only observe the riches getting richer, with no redistribution effect in the network.
%
%
\begin{figure}[t]
\begin{center}
\includegraphics[scale=0.7]{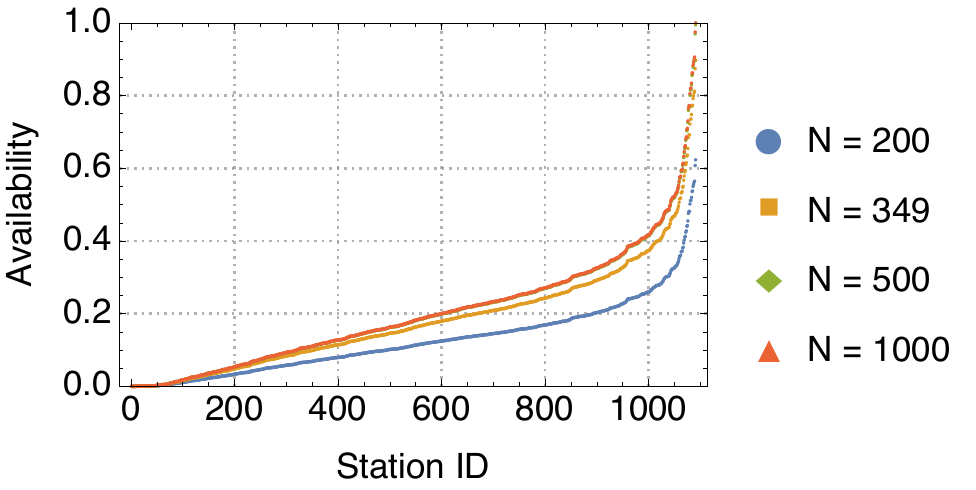}\vspace{-5pt}
\caption{Availability per station with varying fleet size (the curves for N=500 and N=1000 are overlapped).}
\label{fig:availabilityVSstat_fleetsize}
\end{center}\vspace{-20pt}
\end{figure}

The above discussion confirms the intuition that redistribution strategies altering the flows in the networks are needed for real-life car sharing systems. While there exists several proposals in the literature as far as traditional car sharing systems are concerned, there are no results for innovative car sharing systems with stackable vehicles as those described in Section~\ref{sec:intro}. In order to fill this gap, in Section~\ref{sec:relocation_queue} we discuss how to modify the single-server queues of the closed car sharing network to include these stackable capabilities, and how to exploit this new type of queue for setting up a theoretical model for the relocation with stackable vehicles. But before doing this, we make the case for user-based relocation with stackable cars by showing, in the following section, simulation results using two simple heuristics. 
%
%
%
%
%
\section{User-based Relocation: Preliminary Results}\label{sec:relocation_sim}
\noindent
User-based relocation policies are typically considered more convenient for the CS operator than operator-based ones as they do not require the use of dedicate workforce. However, users tend to move accordingly to the flows that are the cause of unbalance in the system in the first place. For this reason, most works on user-based relocation focus on finding the right incentives for users to slightly modify their behaviours in a way more favourable to the CS needs. However, it is still uncertain whether users would be willing to participate in a rebalancing program by accepting to drop off the vehicle at an alternative destination or to pick up a more distant vehicle~\cite{Herrmann2014}. 
One of the main advantages of using stackable cars is that we may not need to change customers' travel behaviours because we can amplify -- by asking the customer to drive a train for relocation -- the ``weak signal" of customers belonging to those flows that go in the right directions for relocation. 

In order to preliminarily assess the impact of vehicle stackability on the relocation performance, in the following we evaluate two simple approaches: $i)$ a \emph{uniform} relocation strategy in which each customer takes a second vehicle to his/her intended destination with a fixed probability $\alpha$; and $ii)$ a \emph{backpressure} strategy in which a customer takes a second vehicle to his/her intended destination only if the number of parked vehicles at the destination station is smaller than at the origin station\footnote{This strategy is inspired by the backpressure routing algorithm, a method for directing traffic around a queueing network that achieves maximum network throughput~\cite{tassiulas1992stability}.}. The rationale behind the latter strategy is to use the redistribution to equalise the queue backlogs (i.e., the number of cars at each station waiting for customers to pick them up). 

\begin{figure}[t]
\begin{center}
\includegraphics[scale=0.7]{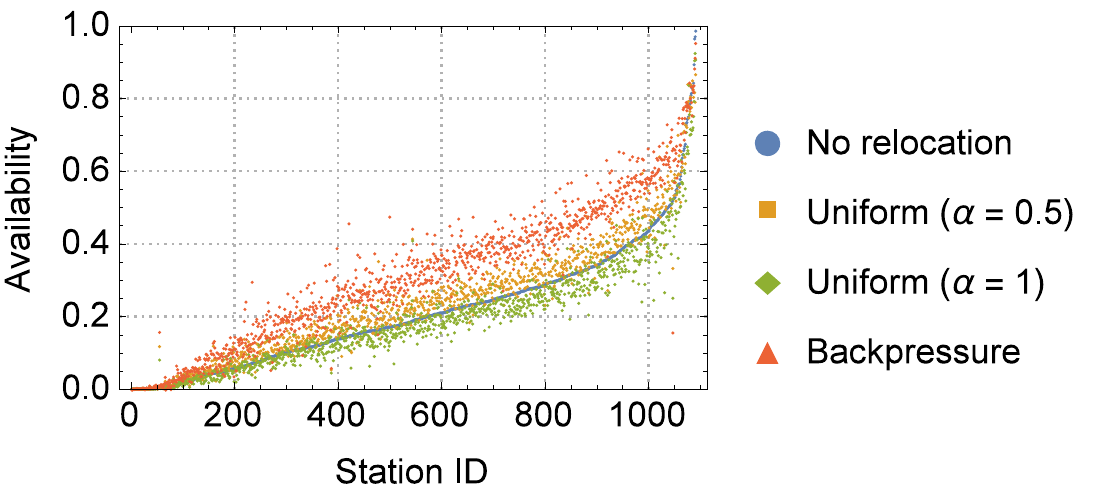}\vspace{-10pt}
\caption{Availability per station, with and without relocation.}
\label{fig:availability_relocation}
\end{center}\vspace{-10pt}
\end{figure}

We simulate the car sharing system using a custom C++ simulator that tracks the evolution of the closed queueing network system described in Section~\ref{sec:queueing_model}. For our evaluation, we use the same configuration as in Section~\ref{sec:fleet}, i.e., cell side length 250m and a large fleet size of 1000~vehicles. The transient period before the system reaches a steady-state is discarded. We compare the relocation strategies in terms of the availability of vehicles that they can provide to each station (Fig.~\ref{fig:availability_relocation}). Important observations can be derived from the results. First, relocation has a complex impact on the car availability. While increasing the fleet size produces a smoothed ``scaling'' effect on the car availability, relocating vehicles may cause stations with similar initial car availabilities to experience a different gain (corresponding to the dispersion of values in Fig.~\ref{fig:availability_relocation}). Second, uniform relocation is not bringing about any significant performance gain, with most of the stations even having a reduced availability. On the contrary, a backpressure relocation scheme is effective in improving car availabilities since it smartly relocates vehicles where there might be a shortage. Third, with a backpressure relocation strategy the performance gain is higher for stations with low-to-medium availabilities. To quantify this trend, in Fig.~\ref{fig:availability_backpressure_relocation} we show the availability variation (in percentage) for each station when using backpressure relocation. The results show that that are stations with very low availabilities that experience an availability increases up to $\sim 200\%$. However, the gains are highly variable and there are also stations that can suffer from degradation of car availability. This motivates the need for design and modelling tools that can allow to calculate optimal relocation probabilities between pairs of stations depending on their service characteristics.  

\begin{figure}[t]
\begin{center}
\includegraphics[scale=0.5]{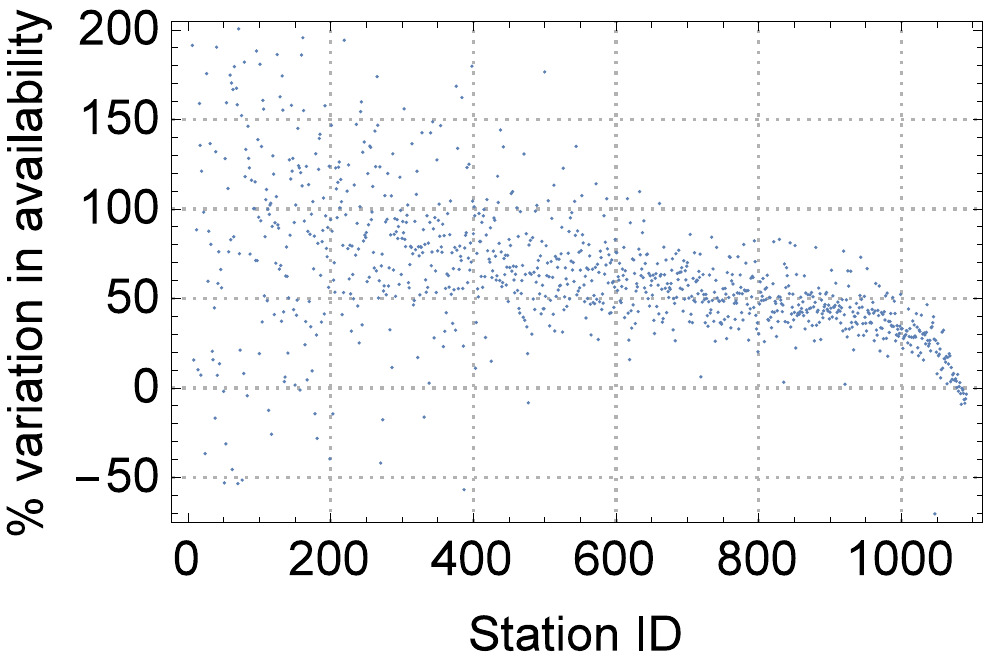}\vspace{-5pt}
\caption{Availability variation per station with backpressure relocation.}
\label{fig:availability_backpressure_relocation}
\end{center}\vspace{-15pt}
\end{figure}

%
%
%
\section{Queues and stackable vehicles}\label{sec:relocation_queue}
\noindent
In order to fill the gap between heuristic-based relocation and optimal relocation (i.e., a relocation strategy that explicitly maximizes some utility function of the CS system variables) we need to be able to represent mathematically the CS system. Queueing theory has been already explored in the literature to design smart relocation strategies for legacy car sharing~\cite{rr16_queue}. However, those mathematical models are not fit to capture car sharing systems that allow for stackable vehicles. Thus, in this section, we develop a new queueing model that addresses this issue. We call this new queue \emph{relocation queue}. A relocation queue can be used, with some modifications, for both user-based and operator-based relocation. Due to its potential impact on CS costs (customers are cheaper than a dedicated workforce) and its many challenges, in this work we focus on user-based relocation.

As a first step, we consider the queue in isolation. As with legacy car sharing systems, customers arrive at stations and pick up vehicles for their journeys. With stackable vehicles, however, a customer can pick up up to $n$ vehicles, in case a redistribution is needed by the CS operator. For the sake of example, in the following we set $n=2$. As highlighted in Section~\ref{sec:relocation_sim}, customers should not be always requested to perform relocation, though, because, e.g., if all customers headed towards service centre $j$ would always relocate vehicles, we might generate too large a flow towards station $j$, possibly emptying the origin station $i$. Hence, the relocation process from station $i$ to station $j$ is regulated by parameter $\alpha_{ij}$, which can be seen as the probability for customers headed for station $j$ from station $i$ to pick up an additional vehicle for relocation.

Let us now take a step back and rethink how the service rate has been modelled so far for single-server queues in Section~\ref{sec:queueing_model}. There, we have used a unique service rate $\mu_i$ for each station~$i$, regardless of the destination of customers picking up vehicles at station $i$. This approach is not suitable anymore, because, with the relocation queue, we want to distinguish between customers heading for different destinations, since they may operate redistribution with different probabilities $\alpha_{ij}$. Thus, we need to explode $\mu_i$ into its components $\mu_{ij}$, each describing the rate at which customers headed for station $j$ arrive at station $i$. Please note that, thanks to the properties of the exponential distribution, $\mu_i = \sum_j \mu_{ij}$. In fact, the interdeparture interval when the queue is busy is exponentially distributed, hence the arrival of customers headed towards different destinations can be handled as if it were a superposition of Poisson processes. 

%
%
%
\subsection{An approximation of the relocation queue\label{sec:relocation_queue_approx}}
\noindent
An exact representation of what happens at a relocation queue should rely on the concept of \emph{batch queue}~\cite{bolch2006queueing}. However, since dealing with batch queues can be complicated when linking together stations into a closed networking system, in this section we propose a modified relocation queue, based on two approximations that allow us to significantly reduce the complexity of the model. The first approximation relies on the intuition that, when customers relocate vehicles, it is like they were picking up vehicles twice as fast. More in detail, customers headed for station~$j$ seem to arrive twice as fast at station $i$ with probability~$\alpha_{ij}$. We can represent this ``modified'' service, say $S_{ij}$, with a mixture distribution, where random variable $S_{ij}$ is given by the following:
\begin{equation}\label{eq:s_ij}
S_{ij} = (1 - \alpha_{ij}) \textrm{Exp}\left(\mu_{ij}\right) + \alpha_{ij} \textrm{Exp}\left(2 \mu_{ij} \right).
\end{equation}
The above equation basically express the concept that, with relocation (i.e., with probability $\alpha_{ij}$) the service is twice as fast, while, without relocation (which happens with probability $1-\alpha_{ij}$), the service runs at its unmodified rate. Please note that the service described by  $S_{ij}$ only holds when there are at least 2 vehicles at station $i$. Otherwise, the customers pick up vehicles with their unmodified rate~$\mu_{ij}$ with probability $1$.

The distribution of $S_{ij}$ is not exponential but, in order to keep the analysis tractable, we want to approximate it with an exponential random variable. To this aim, we simply compute the expectation of $S_{ij}$ (equal to $\mathbb{E}[S_{ij}] = \frac{2 - \alpha_{ij}}{2 \mu_{ij}}$) and use its inverse as the rate $\gamma_{ij}$ of the approximating exponential random variable. This is the second approximation for the modified relocation queue.
Now that we have exponential service times for each destination $j$, we can again compute the overall service rate by summing the service rates for each destination. The modified relocation queue is illustrated in Fig.~\ref{fig:modified_reloc_queue}. This queue belongs to the category of \emph{load-dependent} queues, since the service rate is dependent on the current state of the queue.

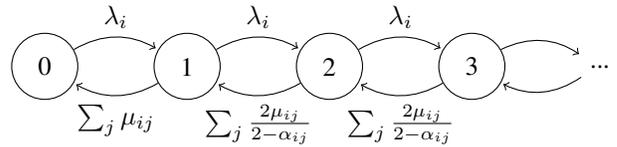
\begin{figure}
\begin{center}
\begin{tikzpicture}
        \node[state]             (0) {0};
        \node[state, right=of 0] (1) {1};
        \node[state, right=of 1] (2) {2};
        \node[state, right=of 2] (3) {3};
        \node[draw=none, right=of 3]           (4) {...};
 
        \draw[every loop]
            (0) edge[bend left, auto=left] node {$\lambda_i$} (1)
            (1) edge[bend left, auto=left] node {$\sum_j \mu_{ij}$} (0)
            (1) edge[bend left, auto=left] node {$\lambda_i$} (2)
            (2) edge[bend left, auto=left] node {$\sum_j \frac{2 \mu_{ij}}{2 - \alpha_{ij}}$} (1)
            (2) edge[bend left, auto=left] node {$\lambda_i$} (3)
            (3) edge[bend left, auto=left] node {$\sum_j \frac{2 \mu_{ij}}{2 - \alpha_{ij}}$} (2)
            (3) edge[bend left, auto=left] node {} (4)
            (4) edge[bend left, auto=left] node {} (3);
    \end{tikzpicture}
\end{center}\vspace{-10pt}
\caption{Markov chain for the modified relocation queue}
\label{fig:modified_reloc_queue}\vspace{-15pt}
\end{figure}

The stationary distribution of the modified relocation queue can be found by solving the balance equations 
using, e.g., the difference equation technique~\cite{medhi2002stochastic}. Due to space limitations we do not report the full mathematical derivations, which are tedious, but only the final results. Specifically, we find that the system has a solution only if $\zeta = \frac{\lambda_i }{\sum_j \frac{2 \mu_{ij}}{2-\alpha_{ij}}}$ is strictly smaller than $1$. This is the new equilibrium condition, which replaces $\frac{\lambda}{\mu_i} < 1$ for the M/M/1 queue. At equilibrium, the stationary distribution of the queue is then given by the following:
\begin{equation}\label{eq:approx_queue_stat_distr}
\begin{cases}
    \pi_0 = \frac{1-\zeta }{-\zeta +\rho +1} \\
    \pi_n = \pi_0 \; \rho \; \zeta^{n-1} & n \geq 2
\end{cases}
\end{equation}
where we have denoted with $\rho$ the quantity $\frac{\lambda}{\mu_i}$.


%
%
%
\subsection{Validation\label{sec:relocation_queue_approx_validation}}
\noindent
In order to validate the proposed approximate model for the relocation queue, in the following we compare important metrics, such as the utilization and the expected number of available cars obtained with this model against simulation results (thus, exact) of the relocation queue. In order to make the discussion easy to follow, we consider $\alpha_{ij} = \alpha_i, \forall j$ (basically, we assume that the probability of relocation is the same for all destinations).

Fig.~\ref{fig:utilization_approx_queue} shows the utilization for different $\alpha_i$ values. We observe that the biggest discrepancies appear for intermediate values of $\alpha_i$, but that, overall, the approximation is very close to the exact characterization of the relocation queue. In addition, the error is always greater for larger values of $\rho$. This is due to the fact that, when $\rho$ is small, the queue is in light traffic, with typically 0 or 1 cars, and thus relocation cannot be performed most of the times. Since the approximation applies to transitions between states with at least two cars, it has little effect in this situation. Vice versa, with larger $\rho$ the chances for relocation are higher, thus the approximation on the states with more than 2 vehicles starts to kick in. Please note, however, that the difference between the exact and approximate models is not much. Analogous conclusions can be drawn when looking at the difference between the predictions of the approximate model and simulations  as far as the average number of cars parked at the station is concerned (Fig.~\ref{fig:cars_parked_approx_queue}).

\begin{figure}[t]
\begin{center}
\includegraphics[scale=0.5]{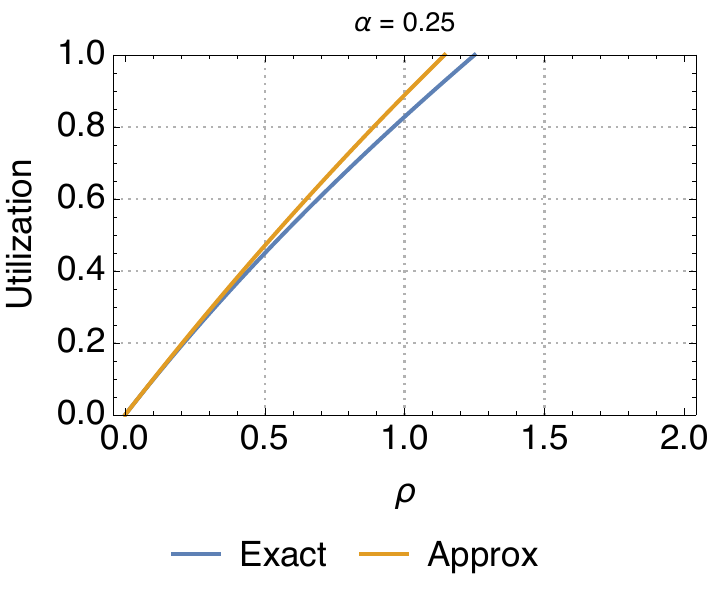}
\includegraphics[scale=0.5]{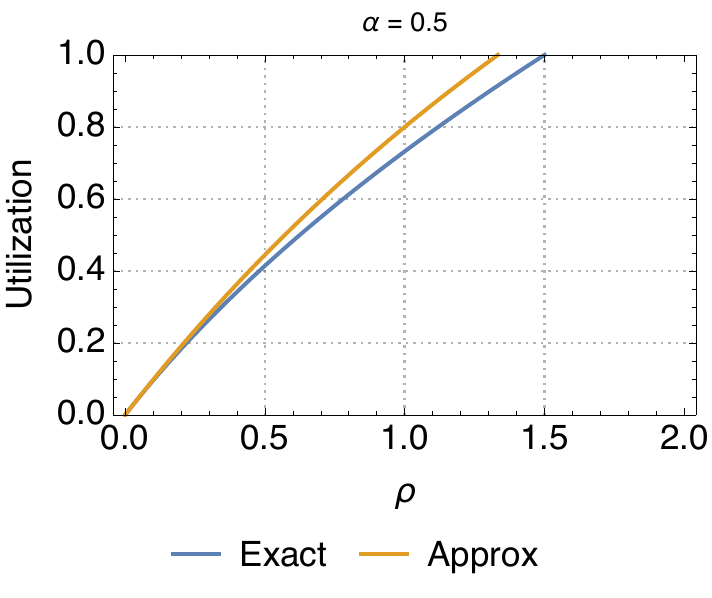}
\includegraphics[scale=0.5]{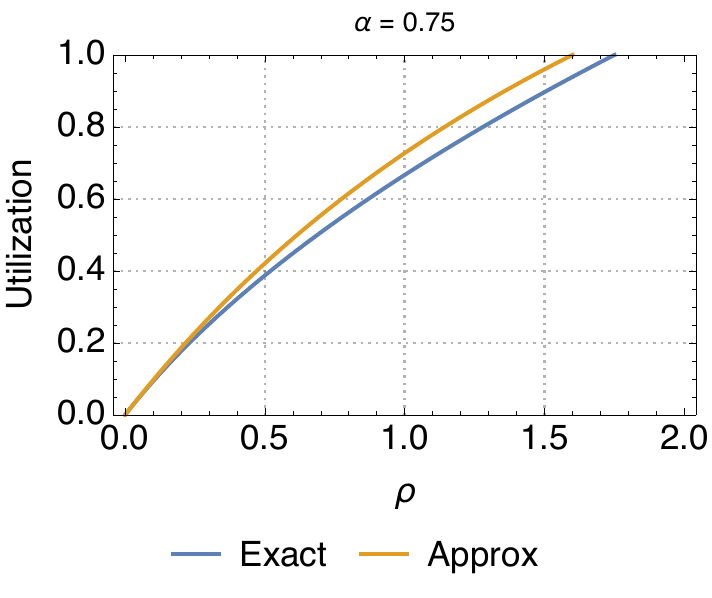}
\includegraphics[scale=0.5]{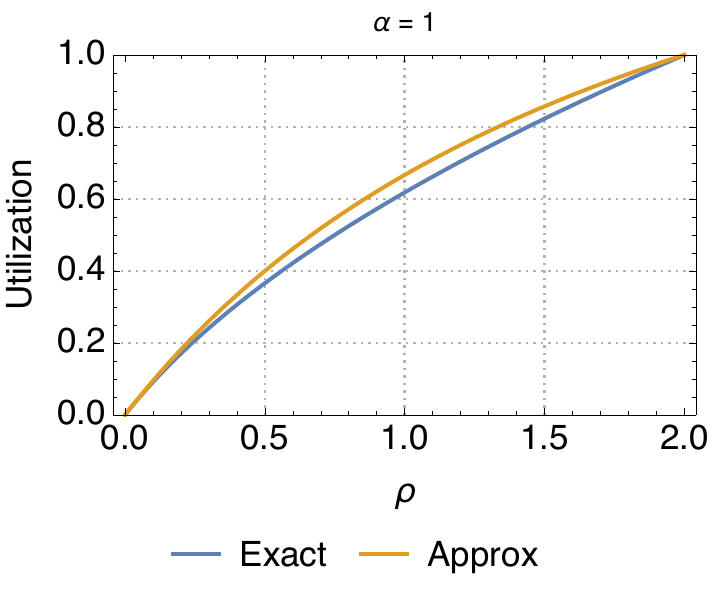}\vspace{-5pt}
\caption{Utilization: simulation VS approximate model}
\label{fig:utilization_approx_queue}
\end{center}\vspace{-15pt}
\end{figure}

\begin{figure}[t]
\begin{center}
\includegraphics[scale=0.5]{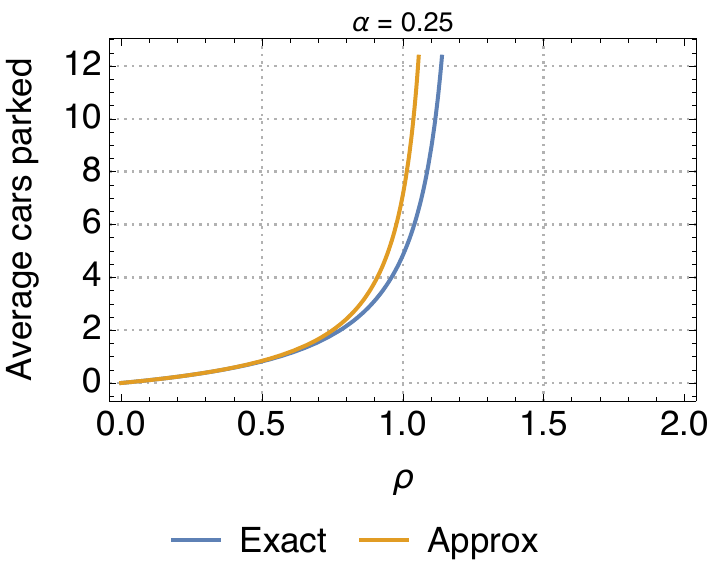}
\includegraphics[scale=0.5]{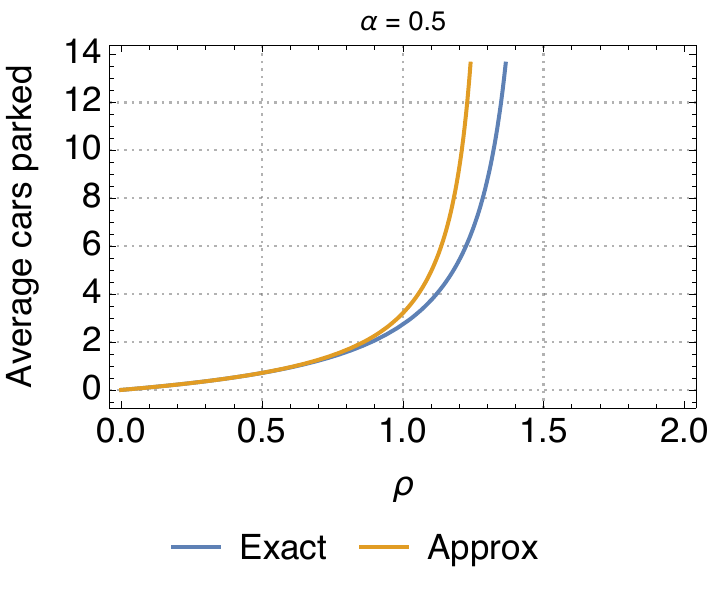}
\includegraphics[scale=0.5]{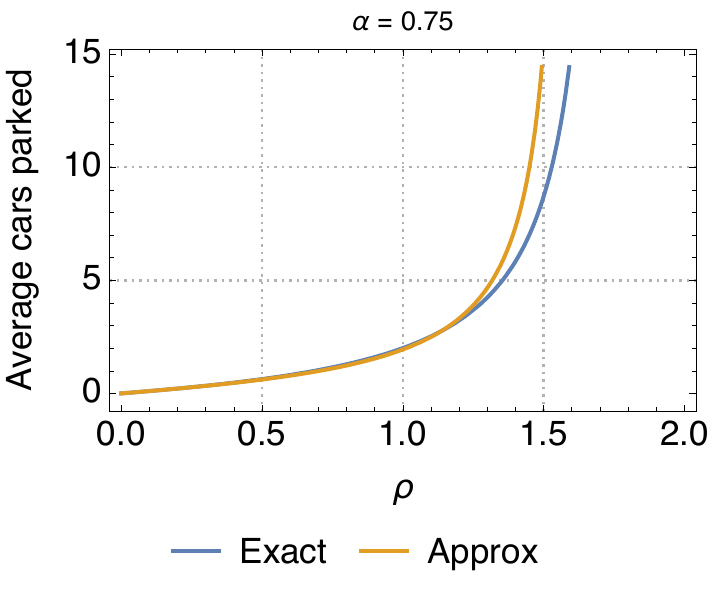}
\includegraphics[scale=0.5]{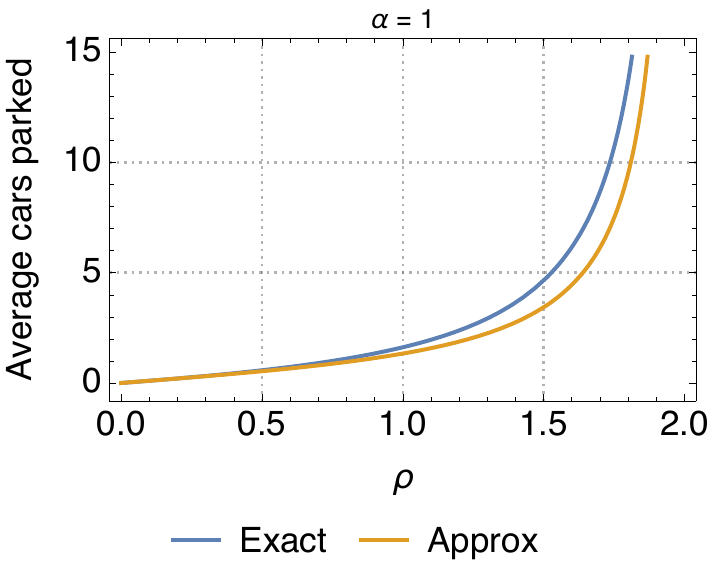}\vspace{-5pt}
\caption{Average number of cars parked at the station: simulation VS approximate model}
\label{fig:cars_parked_approx_queue}
\end{center}\vspace{-20pt}
\end{figure}

\subsection{Derivation of the routing probabilities}
\label{sec:relocation_queue_approx_routing}

So far we have considered the queue in isolation. In order to build a network of relocation queues we need to be able to compute the routing probability when vehicles leave a station $i$. Without relocation, the probability that  an idle vehicle $v$ is next picked up by a customer headed for station $j$ is clearly dependent on the arrival process of customers at station $i$. Since each arrival is exponentially distributed with rate $\mu_{ij}$, by the properties of the exponential distribution, the probability that vehicle~$v$ is picked up by a customer whose destination is station $j$ is equal to $\frac{\mu_{ij}}{\mu_i}$, where $\mu_i = \sum_j \mu_{ij}$.
Relocation effectively alters this probability, since some customers may pick up two vehicles instead of one. We derive the new routing probabilities in Lemma~\ref{lemma:routing} below.

\begin{lemma}\label{lemma:routing}
The probability that a vehicle leaving relocation queue $i$ goes to another queue $j$ is given by the following:
\vspace{-10pt}
\begin{equation}\label{eq:approx_queue_routing_prob}
p_{ij} = \frac{\mu_{ij}}{\mu_i} + \lambda_i \left( \frac{\gamma_{ij}}{\gamma_i^2}-\frac{\mu_{ij}}{\gamma_i \mu_i} \right),
\end{equation}
where for convenience of notation we have defined $\gamma _{ij} = \frac{2 \mu _{ij}}{2-\alpha _{ij}}$ (corresponding to the rate of $S_{ij}$ in Equation~\ref{eq:s_ij}) and $\gamma_i = \sum _{j} \frac{2 \mu_{ij}}{2-\alpha_{ij}}$.
\end{lemma} 

\begin{proof}
First of all, since we are only interested in studying what happens when the server is busy, we rescale the stationary probability excluding case $n=0$. We obtain $\xi _n=\frac{\pi _n}{1-\pi _0}$. Now the computation of the routing probability $p_{ij}$ is straightforward: destination $j$ is selected for vehicle $v$ with the probability that a  customer headed for station $j$ arrives before the others, i.e.,  $\frac{\gamma_{ij}}{\gamma_i}$ when there are $n \geq 2$ jobs in the system (which happens with probability $1 - \xi_1$), $\frac{\mu_{ij}}{\mu_i}$ otherwise. We can write the above as $p_{ij}=\left(1-\xi _1\right) \frac{\gamma_{ij}}{\gamma _i} + \xi_1 \frac{\mu_{ij}}{\mu_i}$. By simply substituting Equation~\ref{eq:approx_queue_stat_distr} in the expression for $\xi_1$, we can express $p_{ij}$ as in Equation~\ref{eq:approx_queue_routing_prob}.
\end{proof}

In practice, Lemma~\ref{lemma:routing} tells us that we can increment the baseline routing probability ($\frac{\mu_{ij}}{\mu_i}$, the one without relocation) by a quantity $\lambda_i \left( \frac{\gamma_{ij}}{\gamma_i^2}-\frac{\mu_{ij}}{\gamma_i \mu_i} \right)$. The term within the parenthesis depends only on the customers' arrival processes and the configured relocation probability. Term $\lambda_i$ corresponds to the input traffic. The higher the input traffic, the higher the impact of a relocation policy at a station. Vice versa, if $\lambda_i$ is small, even with $\alpha_{ij}=1$ we cannot increase significantly the traffic towards station $j$. 
Another interesting finding from Equation~\ref{eq:approx_queue_routing_prob} is that the terms inside the parenthesis becomes zero when $\alpha_{ij} = \alpha_i$. So, when the relocation probability is set to the same value for all destinations, it has no effect on the routing probability. This means that relocation should always favour one destination over the others, in order to make an impact on the system. Note that this analytical results confirm what observed in Fig.~\ref{fig:availability_relocation} for the uniform relocation scheme. 

But how big an impact user-based relocation can have on the system? For the sake of example, in the following we consider a relocation queue $i$ with two destinations. The customers headed for destination 1 are slower than those headed for destination 2, or, in other words, they arrive at a lower rate. We set $\mu_{i1} = 0.2$ and $\mu_{i2} = 0.8$. Then, we plot how the routing probability varies when we change the relocation probability. In Fig.~\ref{fig:routing_prob_alpha1} we assume we want to relocate towards destination 1 (the slow one), and we vary the relocation probability $\alpha_{i1}$ in $\{0, 0.4, 0.8, 1\}$. In Fig.~\ref{fig:routing_prob_alpha2} we assume we want to relocate towards destination 2 (the fast one), and we vary the relocation probability $\alpha_{i2}$ again in $\{0, 0.4, 0.8, 1\}$. We observe that when an heavy flow of customers tries to take vehicles for relocation from a small flow, the resulting routing probability is not altered significantly (Fig.~\ref{fig:routing_prob_alpha2}). The greatest effect that relocation can have is reached when the opposite happens, i.e., when a small flow competes with a heavier flow (Fig.~\ref{fig:routing_prob_alpha1}). This feature should be taken into account when designing user-based relocation strategies with stackable vehicles. Using the formulas derived in this section, we can even obtain a stronger result, summarised in Theorem~\ref{theo:bound_relocation}.

\begin{figure}[t]
\begin{center}
\includegraphics[scale=0.7]{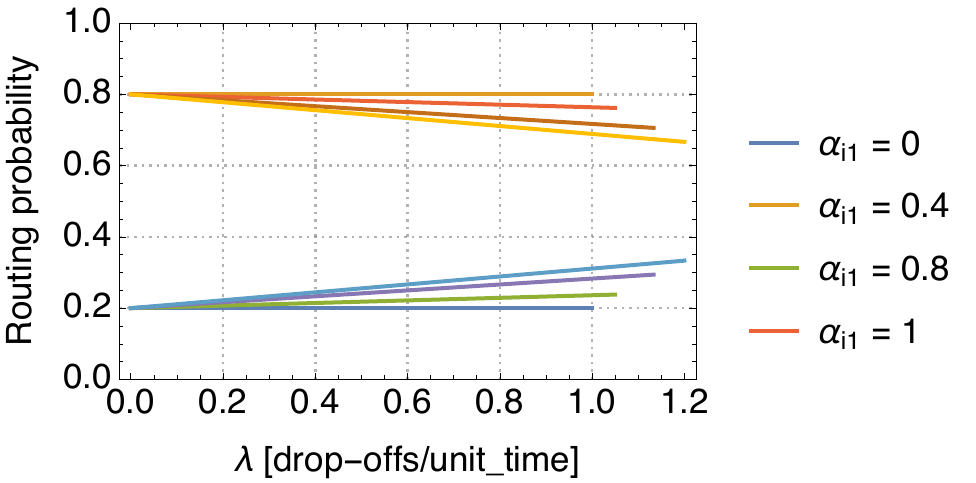}\vspace{-5pt}
\caption{Routing probability: relocation to the station with the slowest customers.}
\label{fig:routing_prob_alpha1}
\end{center}\vspace{-20pt}
\end{figure}

\begin{figure}[t]
\begin{center}
\includegraphics[scale=0.7]{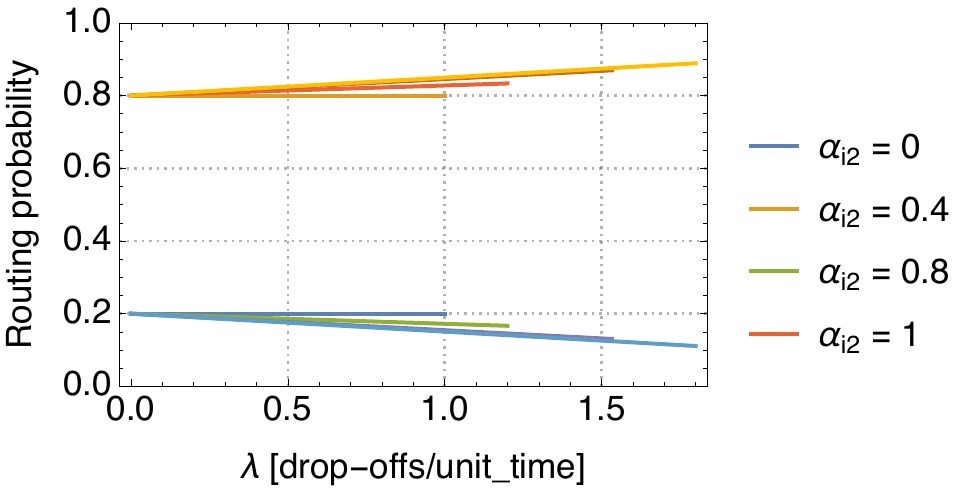}\vspace{-5pt}
\caption{Routing probability: relocation to the station with the fastest customers.}
\label{fig:routing_prob_alpha2}\vspace{-20pt}
\end{center}
\end{figure}

\begin{theorem}[Upper-bound on user-based relocation]\label{theo:bound_relocation}
User-based relocation can increase the routing probability from station $i$ towards another station $j$ (for which the initial $p_{ij}$ was greater than 0) by at most $3 - 2 \sqrt{2}$ ($\sim 0.17$). Hence, the flow of vehicles from $i$ to $j$ can never grow more than $17\%$.
\end{theorem}

\begin{proof}
It is easy to show that the maximum improvement in the routing probability is obtained when all relocation efforts are made towards a single station. Let us call this station station $z$. We have $\alpha_{iz} = 1$, $\alpha_{ij} = 0, \forall j \neq i,z$. When this holds, we have $\gamma_i = \mu_i + \mu_{iz}$.  We can then rewrite Equation~\ref{eq:approx_queue_routing_prob} as follows:
\begin{equation}
p_{iz} = \frac{\mu_{iz}}{\mu_i} + \frac{\lambda_i}{\gamma_i} \left( \frac{2 \mu_{iz}}{\mu_i + \mu_{iz}}-\frac{\mu_{iz}}{\mu_i} \right).
\end{equation}
Quantity $\frac{\lambda_i}{\gamma_i}$ needs to be smaller than $1$ for the queue to be stable. Forcing the constraint that $\mu_{iz}$ and $\mu_{i}$ should be strictly greater than zero, the maximum of the function within the parenthesis yields the result $3 - 2 \sqrt{2}$. Since the output flow from station $i$ to station $z$ at equilibrium is simply $\lambda_i p_{iz}$, we have that the flow of vehicles from $i$ to $j$ can be at most increased by $17\%$.
\end{proof}

Clearly, since the incoming flow to station $j$ is the composition of all the flows arriving at $j$, the overall impact (at $j$) of relocation can be much larger than $17\%$. What the bound tells us is that each incoming flow at $j$ can grow at most by $17\%$ with user-based relocation. In order to understand under which circumstances this bound is attainable, the next step in this research will focus on the study of a closed network of \emph{relocation} queues.

%
%
%
\section{Conclusions}\label{sec:conclusions}
\noindent
In this work we have considered the problem of relocating shared vehicles in a car sharing system in which cars can be driven in a train. Building upon the related literature, we have recalled the commonly used closed queueing network model for the characterisation of legacy car sharing systems. Next, we have validated this model against a real car sharing trace. To the best of our knowledge, this is the first trace-based validation of queueing theoretical models for car sharing systems. Our validation has shown that these models are extremely promising for modelling CS systems, as they are able to predict their behaviour quite accurately. Then, we have made the case for user-based relocation showing how simple strategies can be enough to improve car availabilities. Finally, we have proposed a new type of queue, to be used for modelling a station in which vehicles can be driven in a train. We have validated this model, showing its accuracy, and we have used it for deriving an upper-bound on the relocation performance: even in the best case, the flow of vehicles from a station $i$ to another station  $j$ can never grow more than $17\%$. 

As future work we plan to investigate the behaviour of a network of relocation queues, and to exploit our queueing theoretical framework to derive optimal relocation policies. Furthermore, our model could be extended to also characterise operator-based relocation strategies in which a dedicated workforce is used to relocate train of vehicles, which we believe to be of both theoretical and practical importance.

\vspace{-10pt}

%
%
%
%
%
%
%
%
%
%


\bibliographystyle{IEEEtran}

\bibliography{carsharing.bib}

\end{document}